\documentclass[11pt]{article}

\usepackage{amsfonts,amsmath,amssymb}
\usepackage{url}
\usepackage{tikz}
\usetikzlibrary{arrows}
\usepackage{multirow}
\usepackage{xspace}
\usepackage{amsthm}
\usepackage{comment}
\usepackage[algo2e,vlined,linesnumbered,tworuled]{algorithm2e}
\usepackage[english]{babel}

\SetAlFnt{\footnotesize}
\SetAlCapFnt{\footnotesize}
\SetAlCapNameFnt{\footnotesize}
\SetVlineSkip{1pt}
\SetAlCapHSkip{0pt}

\newcommand{\MICM}{{\textsc{IMA}}\xspace}

\newcommand{\CICM}{{\textsc{CostIMA}}\xspace}
\newcommand{\Greedy}{{\textsc{Greedy}}\xspace}

\newcommand{\MSC}{{\textsc{MSC}}\xspace}

\newtheorem{theorem}{Theorem}
\newtheorem{lemma}[theorem]{Lemma}

\title{Recommending links through influence maximization}
\author{Gianlorenzo D'Angelo,$^{1}$ Lorenzo Severini,$^{2}$ Yllka Velaj$^{1,3}$\\
\normalsize{$^{1}$Gran Sasso Science Institute (GSSI),}\\
\normalsize{$^{2}$ISI Foundation,}\\
\normalsize{$^{3}$University of Chieti-Pescara.}\\
}
\date{}

\begin{document}
\maketitle
\begin{abstract} 
The link recommendation problem consists in suggesting a set of links to the users of a social network in order to increase their social circles and the connectivity of the network.
Link recommendation is extensively studied in the context of social networks and of general complex networks due to its wide range of applications. 
Most of the existing link recommendation methods estimate the likelihood that a link is adopted by users and recommend links that are likely to be established. 
%Other approaches recommend links that, when added to the network, increase the importance of a user in a network, where the importance of a node is determined by means of a range of structural properties called centrality indices.
However, most of such methods overlook the impact that the suggested links have on the capability of the network to spread information.  Indeed, such capability is directly correlated with both the engagement of a single user and the revenue of online social networks.

In this paper, we study link recommendation systems from the  point of view  of information diffusion. In detail, we consider the problem in which we are allowed to spend a given budget to create new links so to suggest a bounded number of possible persons to whom become friend in order to maximize the influence of a given set of nodes. We model the influence diffusion in a network with the popular Independent Cascade model.

%We show that the problem cannot be approximated within a constant factor greater than $1-(2e)^{-1}$, unless $P=NP$. We propose then two approximation algorithms: the former one works for a unit-cost version of the problem where each edge has unit-cost; the latter algorithm is most general one where we consider arbitrary costs for adding new links. The algorithms guarantee an approximation factor of $1-(e)^{-1} - \epsilon$, where $\epsilon$ is any positive real number.
\end{abstract}

\section{Introduction}

Link recommendation is one of the most important features of online social networking sites. Recommending suitable connections to social network users has a twofold impact: it improves the user's experience by enlarging users' social circles and connectivity, and it increases social network revenue by enhancing the user engagement and retention rate.

Most of the existing link recommendation methods estimate the likelihood that a link is adopted by users and recommend links that are likely to be established~\cite{BL11,LFS15,LK03,LZ11}.
In social networks such likelihood is estimated by considering the similarity of user profiles and some structural network properties.
Friendship social networks, like Facebook, exploit similarity metrics that are based on the number of common neighbors. For example, the Friend-of-Friend (FoF) algorithm~\cite{LK03} recommends the users that have the highest number of common friends with the receiver of the recommendation. Other examples of such similarity metrics are the Adamic-Adar~\cite{AdAd} index, the Jaccard's coefficient~\cite{jaccard}, and the preferential attachment index~\cite{BBCR03,N01}. In content-centric social networks such as Twitter and Google+ the link recommender systems  take also into account the similarity of the users' interests.

Most of the known methods aim at having a high accuracy in the prediction of the suggested links, without considering the impact of the new links on the capability of the network to spread information. This approach speeds up the network growth but is only able to infer links that will likely occur in the near future.
% does not increase the number of user's interactions that are eventually formed, since the suggested links are very likely to occur in the near future without the help of recommender systems.
Another drawback of the existing methods is that, in most of the cases, they suggest links to a short range of users and this does not necessarily lead to a network growth and an improved user engagement~\cite{YWBWWC15}.

On the other hand, the capability of a social network to spread information is directly correlated with both the engagement of a single user and the revenue of an online social network~\cite{CRRB12}. From a user's perspective, being able to quickly and effectively disseminating information is highly desirable as it helps the user to share contents so to reach a large number of other users. This in turn helps the user to build its own social reputation, to express and diffuse its own opinion, and to discover novel contents and information. From the social network point of view, the effectiveness of the information spreading capabilities helps in improving the user engagement, which in turn increases the retention rate and the number of new subscriptions. Moreover, being able to quickly deliver diverse contents to a large portion of the users, increases the opportunities of making revenue from advertisement.

Most of the existing link recommendation systems overlook this aspect, which is crucial for both users and social networks. In this paper, we consider the problem of recommending links to a given set of users in a social network, without exceeding a given budget, in such a way that the number of users reached by the contents generated by such set of users is maximized. Our main objective is to improve the capability of the given users to diffuse their own contents, this in turn drives the network evolution towards an increment of the spreading capability of the whole network.

Several models of information diffusion have been introduced in the literature, two widely studied models are: the \emph{Linear Threshold Model} (LTM)~\cite{G78,KKT15,S06} and the \emph{Independent Cascade Model} (ICM)~\cite{GLM01, GLM01a, KKT03, KKT15}.
In both models, we can distinguish between \emph{active}, or \emph{infected}, nodes which spread the information and inactive ones. At the beginning of the process a small percentage of nodes of the graph is set to active in order to let the information diffusion process start. Such nodes are called \emph{seeds}. Recursively, currently infected nodes can infect their neighbours with some probability.  After a certain number of such cascading cycles, a large number of nodes might becomes infected in the network.
%
% In LTM the idea is that a node becomes active as more of its neighbours become active. More formally, each node $u$ has a threshold $t_u$ chosen uniformly at random in the interval $[0,1]$. The threshold represents the weighted fraction of neighbours of $u$ that must become active in order for $u$ to become active. During the process, a node becomes active if the total weight of its active neighbours is greater than its threshold.
%
% In ICM, an active node $u$ tries to influence one of its inactive neighbours but the success of node $u$ in activating the node $v$ only depends on the propagation probability of the edge from $u$ to $v$ (each edge has its own value). Regardless of its success, the same node will never get another chance to activate the same inactive neighbour.
The process terminates when no further node gets activated.
In this paper we adopt  ICM to model the way in which the contents are propagated in the network.

% In this paper, we assume that we are given a set $A$ of seeds, we want find a fixed number of links to add from nodes in $A$ in order to maximize the expected eventual number of active nodes at the end of the diffusion process.
\subsection{Related works}
The problem of recommending links to the users of a social network has been widely studied, we refer to~\cite{LFS15} and~\cite{LZ11}  for surveys on the link recommendation and link prediction problems, respectively.
The problem of recommending links by taking into account the information spreading capability, instead, has received little attention in the literature. In the following we focus on such problem and on the problems of modifying a graph in order to maximize or minimize the spread of information through a network under LTM and ICM models.

Yu et al.~\cite{YWBWWC15} propose a recommendation algorithm called ACR-FoF (algebraic connectivity regularized friends-of-friends) that uses the algebraic connectivity of a connected network to estimate its capability for spreading contents. The ACR-FoF takes also into account the success rate of the suggested links. The authors give experimental evidence that ACR-FoF  improves the spread of contents in a social networks but do not prove any approximation guarantees.
Chaoji et al.~\cite{CRRB12} consider a model in which each node is associated with an independent probability to share a content with all its neighbors and with a set of contents that are generated by the node itself. The problem they study consists in maximizing the expected number of nodes influenced by some of the contents by adding a set of edges under the constraint that each node has at most $k$ incident new edges. They show that the problem is $NP$-hard and propose an information diffusion model called Restricted Maximum Probability Path Model in which a content is propagated between two users along the path with maximum probability among those containing a recommended edge. They show that under this model the objective function is submodular and hence the problem can be approximated within a constant bound.
Li et al.~\cite{LXLSGS13} introduce the notion of user diffusion degree which is a measure of the influence that a user has and is computed by combining community detection algorithms with information diffusion models. They propose an algorithm that suggests links by combining the diffusion degree with the FoF algorithm and, by means of experiments on two networks, show that it outperforms some known baseline in terms of number of affected nodes under the ICM and LTM models.

To the best of our knowledge, under LTM, the graph modification problems that have been studied are those outlined in what follows. Khalil et al.~\cite{KDS14} consider two types of graph modification, adding edges to  or  deleting  edges  from  the  existing  network to minimize the information diffusion and they show that this network structure modification problem has a supermodular objective and therefore can be solved by algorithms with provable approximation guarantees. Zhang et al.~\cite{ZAVP15} consider arbitrarily specified groups of nodes, and edge and node removal from the groups. They develop algorithms with rigorous performance guarantees and good empirical performance. Kimura et al.~\cite{Kimura08} use a greedy approach to delete edges under the LTM  but do not provide any rigorous approximation guarantees. Kuhlman et al.~\cite{Kuhlman} propose heuristic algorithms for edge removal under a simpler deterministic variant of LTM which is  not  only  hard,  but  also  has  no  approximation  guarantee.  Papagelis~\cite{P15} and Crescenzi et al.~\cite{CDSV16} study the problem of augmenting the graph in order to increase the connectivity or the centrality of a node, respectively and experimentally show that this increases the expected number of eventual active nodes. Parotsidis et al.~\cite{PPT16} study the problem of recommending links with the objective of improving the centrality of a node within a network.

Under ICM, Wu et al.~\cite{WSZ15} consider graph modifications other than edge addition, edge deletion and source selection, such as increasing the probability that a node infects its neighbours. They proved that optimizing the selection of such modifications with a limited budget is $NP$-hard and is neither submodular nor supermodular. Sheldon et al.~\cite{Sheldon} study the problem of node addition to maximize the spread of information, and provide a counterexample showing that the objective  function  is not submodular. Kimura et al.~\cite{Kimura} propose methods for efficiently finding good approximate solutions on the basis of a greedy strategy for the edge deletion problem under the ICM, but do not provide any approximation guarantees.

In this paper, we adopt the independent cascade model and investigate the problem of adding a limited number of edges incident to an arbitrary set of initial seeds, without exceeding a given budget $k$, in order to maximize the spread of information in terms of number of nodes that eventually become active. The problem we analyze differs from above mentioned ones since we make the reasonable restriction that the edges to be added can only be incident to the seed nodes and that to add such edges there is a cost to be paid. To our knowledge, similar problems have never been studied for the independent cascade model. We refer to this problem as the \emph{Cost Influence Maximization with Augmentation problem} (\CICM).

\subsection{Our results}
We first focus on the unit-cost version of the problem that we call \emph{Influence Maximization with Augmentation problem} (\MICM). In such problem the cost of adding  any edge is constant and equal to $1$. We show that \MICM is $NP$-hard to be approximated within a constant factor greater than $1-(2e)^{-1}$ (Section~\ref{sec:inapx}). We then provide an approximation algorithm that almost matches such upper bound by guaranteeing an approximation factor of $1-(e)^{-1} - \epsilon$, where $\epsilon$ is any positive real number (Section~\ref{sec:greedy}). The algorithm is based on a greedy technique and the approximation factor is proven by showing that the expected number of activated nodes is monotonically increasing and submodular with respect to the possible set of edges incident to the seeds.

Then, we study the more general \CICM problem where we are given a budget $k$ and the cost of edges is in $[0, 1]$. We propose an algorithm that combines greedy and enumeration techniques and that achieves an approximation guarantee of  $1-(e)^{-1}$
 (Section~\ref{sec:approx_easy}).

Both the \MICM and \CICM problems are interesting: even though the former represents a limited number of real scenarios, there is a greedy approximation algorithm that guarantees an approximation factor of ${1-(e)^{-1}}$ exploiting the properties of submodular functions. The latter problem, instead, introduces a more flexible and general model for the application to link recommendation even if the greedy heuristic of \MICM can not be trivially generalized to achieve the same approximation for the budgeted version of the problem, we will show how to improve the approximation factor to $1-(e)^{-1}$ using the enumeration technique.
 
\section{Preliminaries} \label{sec:preliminaries}
A social network is represented by a weighted directed graph $G=(V, E, p, c)$, 
where $V$ represents the set of nodes, $E$ represents the set of relationships, $p:V\times V \rightarrow [0, 1]$ is the propagation probability of an edge, that is the probability that the information is propagated from $u$ to $v$ if $(u,v)\in E$, and $c:V\times V \rightarrow [0, 1]$ is the cost of adding an edge to $E$.

In ICM, each node can be either \emph{active} or \emph{inactive}. If a node is active (or infected), then it is already influenced by the information under diffusion, if a node is inactive, then it is unaware of the information or not influenced.
The process runs in discrete steps. At the beginning of the ICM process, few nodes are given the information, they are known as \emph{seed nodes}. Upon receiving the information these nodes become active. In each discrete step, an active node tries to influence one of its inactive neighbours.  The success of node $u$ in activating the node $v$ depends on the propagation probability of the edge $(u, v)$, independently of the history so far. In spite of its success, the same node will never get another chance to activate the same inactive neighbour. The process terminates when no further nodes became activated from inactive state.

We define the influence of a set $A\subseteq V$ in the graph $G$, denoted $\sigma(A, G)$, to be the expected number of active nodes in $G$ at the end of the process, given that $A$ is the initial set of seeds. Given a set $S$ of edges not in $E$, we denote by $G(S)$ the graph augmented by adding the edges in $S$ to $G$, i.e. $G(S) = (V, E \cup S)$. We denote by $\sigma(A, S)$ the influence of $A$ in $G(S)$.

In this paper, given a set of seeds $A$, we look for a set of edges $S$, to be added to $G$, incident to such seeds that maximize $\sigma(A,S)$. We assume that each edge $e\in (V\times V)\setminus E$ can be selected with cost $c_e\in [0,1]$. In detail, the \CICM problem is defined as follows: given a graph $G = (V, E)$, a budget $k$ and a set $A$ of seeds, find a set $S$ of edges such that $S \subseteq (A\times V)\setminus E$, $c(S)\leq k$, and $\sigma(A, S) $ is maximum, where $c(S) = \sum_{e\in S}c_e$. 

Moreover, we consider also the unit-cost version of our problem, we refer to it as the \MICM problem: %In detail, it is defined as follows: given a graph $G = (V, E)$, a set of vertices $A\subseteq V$, and an integer $k$, find a set $S$ of edges such that $S \subseteq (A\times V)\setminus E$, $|S| \leq k$, and $\sigma(A, S)$ is maximum. Notice that 
it is a particular case of \CICM where each edge $e\in (V\times V)\setminus E$ has cost $c_e=1$ and where $c(S) = |S|$. 

We give now some definitions useful to prove our results.
We will use the definition of \emph{live-edge graph} $X=(V, E_X)$ which is a directed graph where the set of nodes is equal to $V$ and the set of edges is a subset of $E$. More specifically, the edge set $E_X$ is given by a edge selection process in which each edge in $E$ is either \emph{live} or \emph{blocked} according to its propagation probability. We can assume that, for each edge $e=(u,v)$ in the graph, a coin of bias $p_e$ is flipped and the edges for which the coin indicated an activation are live, the remaining are blocked. It is easy to show that the information diffusion process is equivalent to a reachability problem in live-edge graphs: given any seed set $A$, the distribution of active node sets after the diffusion process ends is the same as the distribution of node sets reachable from $A$ in live-edge graphs.

We denote by $\chi(S)$ the probability space in which each sample point specifies one possible set of outcomes for all the coin flips on the edges, that is the set of all possible live-edge graphs.
For a node $a\in V$ and a live-edge graph $X$ in $\chi(S)$, let $R(a, X)$ denote the set of all nodes that can be reached from $a$ in graph $X$, that is for each node $v\in R(a,X)$, there exists a path from $a$ to $v$ consisting entirely of live edges with respect to the outcome of the coin flips that generates $X$. Let $R(A, X)=\bigcup_{a \in A} R(a, X)$, then $\sigma(A, S)$ can be computed as 
$\sigma(A, S)=\sum_{X \in \chi(S)} \mathbb{P}[X]\cdot| R(A, X)|$.

Note that, the influence function $\sigma(A, S)$ cannot be evaluated exactly in polynomial time since it has been proven that it is generally $\# P$-complete for the Independent Cascade Model~\cite{CWW10}. However, by simulating the diffusion process sufficiently many times and sampling the resulting active sets, it is possible to obtain arbitrarily good approximations to $\sigma(A, S)$. The next proposition bounds the number of times that the diffusion process must be simulated to obtain a good approximation of   $\sigma(A, S)$.
\begin{theorem}[\cite{KKT15}]
If the diffusion process starting with $A$ on graph $G(S)$ is simulated at least $\Omega(\frac{n^2}{\lambda^2}\ln\frac{1}{\delta})$ times, then the average number of activated nodes over these simulations is a $(1+\lambda)$-approximation to $\sigma(A,S)$, with probability at least $1-\delta$.
\end{theorem}
Therefore, in the rest of the paper we can assume that we can compute $\sigma(A,S)$ within an arbitrary bound. This reflects to an additional factor $1-\epsilon$, for any $\epsilon>0$, to all algorithms presented in this paper. %For the sake of clarity, we omit this factor from the approximation factor of our algorithms.

Given a set of edges $S$, for each graph $X\in\chi(S)$ and subset of edges $T\subseteq S$, we denote by $X^T$ the graph obtained by removing edges in $T$ from $X$. To avoid cumbersome notation, when $X\setminus T = \{e\}$ we denote $X^{e} = X^{\{e\}}$.
Given two feasible solutions $S_1$ and $S_2$, such that $S_2\subseteq S_1$, we denote with $\delta(S_1,S_2)$ the expected number of nodes affected by $S_1$ and not affected by $S_2$, formally:
$\delta(S_1,S_2) = \sum_{X \in \chi(S_1)} \mathbb{P}[X]\cdot\left(|R(A, X)| - |R(A, X^T)|  \right)$, where  $T=S_1\setminus S_2$.

\section{The \MICM problem}
In this section we present our results for the \MICM problem, i.e. the unit-cost version of the \CICM problem. 
\subsection{Hardness of approximation }\label{sec:inapx}
We first show that the \MICM problem does not admit a PTAS, unless ${P=NP}$. The result holds even when $|A|=1$. In the next section, and in Section~\ref{sec:approx_easy}, we give an algorithm that almost matches the following upper bound on approximation.
\begin{theorem}\label{th:inapx}
It is $NP$-hard to approximate \MICM within a factor greater than $1-(2e)^{-1}$ for any $A$ such that $|A|\geq 1$.
\end{theorem}
\begin{proof}
The proof is based on a reduction from the \emph{maximum set coverage problem} (\MSC) which has been shown to be $NP$-hard to approximate within a factor greater than $1-\frac{1}{e}$~\cite{F98}. In detail, in the \MSC problem, we are given a finite set $X$, a finite family $\mathcal{F}$ of subsets of $X$, and an integer $k'$, and we aim at finding a family $\mathcal{F}'\subseteq\mathcal{F}$ such that $|\mathcal{F}'|\leq k'$ and $s(\mathcal{F}') = |\cup_{S_i\in \mathcal{F}'}S_i|$ is maximum.

We follow the scheme of L-reductions~\cite[Chapter~16]{WS11}, since it has been shown that if there is an L-reduction with parameters $a$ and $b$ from maximization problem $\Pi$ to maximization problem $\Pi'$ , and there is an $\alpha$-approximation algorithm for $\Pi'$, then there is an $(1 - ab(1 - \alpha))$-approximation algorithm for $\Pi$~\cite[Chapter~16]{WS11}.
In our specific case, if there exists an $\alpha$-approximation algorithm $A_\MICM$ for \MICM and the following two conditions are satisfied for some values $a$ and $b$:
 \begin{align}
  OPT(I_\MICM)&\leq a OPT(I_\MSC)\label{lreduction:one}\\
  OPT(I_\MSC) - s(S_\MSC) &\leq b \left(OPT(I_\MICM) - \sigma(A,S_\MICM)\right),\label{lreduction:two}
 \end{align}
where $OPT$ denotes the optimal value of an instance of an optimization problem, then there exists an approximation algorithm $A_\MSC$ for \MSC with approximation an factor of  $1-ab(1-\alpha)$. Since it is $NP$-hard to approximate \MSC within a factor greater than $1-\frac{1}{e}$~\cite{F98}, then the approximation factor of $A_\MSC$ must be smaller than $1-\frac{1}{e}$, unless $P=NP$. This implies that $1-ab(1-\alpha) < 1-\frac{1}{e}$ that is, the approximation factor $\alpha$ of $A_\MICM$ must satisfy $\alpha < 1-\frac{1}{abe}$, unless $P=NP$.

 Therefore, in what follows we give a polynomial-time algorithm that transforms any instance $I_\MSC=(X,\mathcal{F},k')$ of \MSC into an instance $I_\MICM=(G,A,k)$ of $\MICM$ and a polynomial-time algorithm that transforms any solution $S_\MICM$ for $I_\MICM$ into a solution $S_\MSC$ for $I_\MSC$ such that the above two conditions are satisfied.

 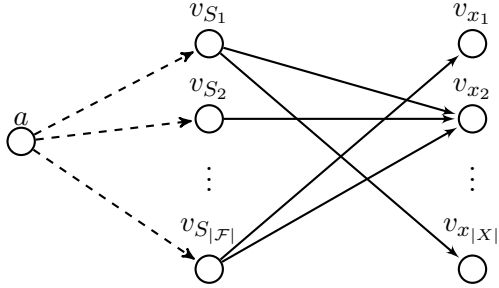
\begin{figure}[t]
 \centering
% \scalebox{0.75}{
  \begin{tikzpicture}[->,>=stealth',shorten >=1pt,auto,node distance=3cm,
  thick,main node/.style={circle,draw}]
  \tikzset{arc/.style = {->,> = latex'}}
  \node (x1) at (6,3.4) {$v_{x_1}$};
  \node[main node] (x11) at (6,3) {};
  
  \node (x2) at (6,2.4) {$v_{x_2}$};
  \node[main node] (x21) at (6,2) {};

  \node (x3) at (6,1.3) {$\vdots$};
%   \node[main node] (x31) at (0,1) {};

  \node (xn) at (6,0.5) {$v_{x_{|X|}}$};
  \node[main node] (xn1) at (6,0) {};

  \node (s1) at (2.5,3.4) {$v_{S_1}$};
  \node[main node] (s11) at (2.5,3) {};
  
  \node (s2) at (2.5,2.4) {$v_{S_2}$};
  \node[main node] (s21) at (2.5,2) {};

  \node (s3) at (2.5,1.3) {$\vdots$};
%   \node[main node] (x31) at (0.5,1) {};

  \node (sn) at (2.5,0.5) {$v_{S_{|\mathcal{F}|}}$};
  \node[main node] (sn1) at (2.5,0) {};

  \node (a) at (0,2) {$a$};
  \node[main node] (a1) at (0,1.7) {}; 
  
  \draw[arc] (s21) to (x21);
  \draw[arc] (s11) to (x21);  
  \draw[arc] (s11) to (xn1);    
  \draw[arc] (sn1) to (x21);   
  \draw[arc] (sn1) to (x11);   
  \draw[->, dashed] (a1) to (s11);
  \draw[->, dashed] (a1) to (s21);
  \draw[->, dashed] (a1) to (sn1);
\end{tikzpicture}
%}
\caption{Reduction used in Theorem~\ref{th:inapx}. The dashed arcs denote those added in a solution.}
\label{fig:inapx}
 \end{figure}
 
Given $I_\MSC=(X,\mathcal{F},k')$, we define $I_\MICM=(G,A,k)$ where $G=(V,E, p)$ is a directed graph containing a node $v_{x_i}$ for each element $x_i\in X$, a node $v_{S_j}$ for each set $S_j\in\mathcal{F}$, a seed node $a$, and a directed edge $(v_{S_j},v_{x_i})$, whenever $x_i \in S_j$. We define $k=k'$ and $A=\{a\}$. We denote by $V_{\mathcal{F}}$ the set of nodes corresponding to sets in $\mathcal{F}$. The propagation probability $p_e$ is equal to 1 if $e\in E \cup A\times V_{\mathcal{F}}$ and 0 otherwise. See Figure~\ref{fig:inapx} for a visualization. Note that in $I_{\MICM}$, the information diffusion is a deterministic process, as all probabilities are 0 or 1.  Moreover, any solution $S$ for $I_\MICM$ contains only edges $(a,v_{S_j})$, for some $S_j \in \mathcal{F}$, since any other edge would have probability 0 of being activated. Given a solution $S_\MICM=\{(a,v_{S_j})~|~S_j\in\mathcal{F}\}$ to $I_\MICM$, we construct the solution $S_\MSC=\{S_j~|~(a,v_{S_j})\in S_\MICM\}$ to $I_\MSC$. 
W.l.o.g., we can assume that $|S_\MICM|=k$ and since by construction $|S_\MSC| = |S_\MICM|$, then $|S_\MSC| = k=k'$.

The nodes influenced by node $a$ in $S_\MICM$ are all nodes $v_{S_j}$ such that $(a,v_{S_j}) \in S_\MICM$ and all the nodes $v_{x_i}$ such that $x_i\in S_j$, for some $S_j$ such that $(a,v_{S_j}) \in S_\MICM$. The nodes of the former type are $k$, while the nodes of the latter type are all those nodes $v_{x_i}$ such that $x_i\in\cup_{(a,v_{S_j}) \in S_\MICM} S_j$. Therefore, $\sigma(A,S_\MICM) = k + |\cup_{(a,v_{S_j}) \in S_\MICM} S_j| = k + s(S_\MSC)$.

It follows that Conditions~(\ref{lreduction:one}) and~(\ref{lreduction:two}) are satisfied for $a=2$, $b=1$ since: $OPT(I_\MICM) =  {OPT(I_\MSC) + k \leq 2OPT(I_\MSC)}$ and ${OPT(I_\MSC) - s(S_\MSC)} = {OPT(I_\MICM) - \sigma(A,S_\MICM)}$, where the first inequality is due to the fact that $OPT(I_\MSC)\geq k$, as otherwise the greedy algorithm finds an optimal solution for \MSC.
The statement follows by plugging the values of $a$ and $b$ into $\alpha < 1-\frac{1}{abe}$. This concludes the proof for $|A|=1$. We can extend to cases in which $|A|>1$ by adding nodes in $A$ that can only form edges with propagation probability equal to 0.
\end{proof}

\subsection{Greedy approximation algorithm}\label{sec:greedy}

In this section, we propose an algorithm that guarantees a constant approximation ratio for the \MICM problem. The algorithm exploits the results of Nemhauser et al. on the approximation of monotone submodular objective functions~\cite{NWF78}. Let us consider the following optimization problem: given a finite set $N$, an integer $k'$, and a real-valued function $z$ defined on the set of subsets of $N$, find a set $S\subseteq N$ such that $|S|\leq k'$ and $z(S)$ is maximum. If $z$ is \emph{monotone and submodular}\footnote{For a ground set $N$, a function $z:2^N\rightarrow \mathbb{R}$ is submodular if for any pair of sets $S\subseteq T \subseteq N$ and for any element $e\in N\setminus T$, $z(S\cup\{e\}) - z(S) \geq z(T\cup \{e\}) - z(T)$.}, then the following greedy algorithm exhibits an approximation of $1-\frac{1}{e}$~\cite{NWF78}: start with the empty set, and repeatedly add an element that gives the maximal marginal gain, that is if $S$ is a partial solution, choose the element $j\in N\setminus S$ that maximizes $z(S\cup\{j\})$.

\begin{theorem}[\cite{NWF78}]\label{th:submodular}
For a non-negative, monotone submodular function $z$, let $S$ be a set of size $k$ obtained by selecting elements one at a time, each time choosing an element that provides the largest marginal increase in the value of $z$. Then $S$ provides a $\left(1-\frac{1}{e}\right)$-approximation.
\end{theorem}

\begin{algorithm2e}[t]
\caption{\Greedy \MICM algorithm.}
\label{alg:greedy}
\SetKwInput{Proc}{algorithm}
%\Proc{\Greedy}
\SetKwInOut{Input}{Input}
\SetKwInOut{Output}{Output}
\Input{A directed graph $G=(V,E)$; a set of vertices $A\subseteq V$; and an integer $k\in\mathbb{N}$}
\Output{Set of edges $S\subseteq (A\times V)\setminus E$ such that $|S|\leq k$}
$S:=\emptyset$\;
\For{$i=1,2,\ldots,k$}
{
  $\hat{e} = \arg\max\{ \sigma(A,S\cup \{e\})~|~e=(a,v) \in (A\times V)\setminus (E\cup S)\}$\;
  $S := S \cup \{\hat{e}\}$\;
}
\Return $S$\;
\end{algorithm2e}
In this paper, we exploit such results by showing that $\sigma(A,\cdot)$ is monotone and submodular with respect to the possible set of edges incident to nodes in $A$. In fact, assuming that for a set $S\subseteq (A\times V)\setminus E$ we are able to compute $\sigma(A,S)$,\footnote{We showed in Section~\ref{sec:preliminaries} how to find an arbitrarily good approximation of $\sigma(A,S)$ in polynomial time.} then algorithm~\ref{alg:greedy} provides a  $\left(1-\frac{1}{e}\right)$-approximation. algorithm~\ref{alg:greedy} iterates $k$ times and, at each iteration, it adds to an initially empty solution $S$ an edge $\hat{e}=(\hat{a}, \hat{v})$ s.t. $(\hat{a}, \hat{v})\in (A\times V)\setminus E$ that, when added to $S$, gives the largest marginal increase in the value of $\sigma(A,S)$, that is $\sigma(A,S\cup \{\hat{e}\})$ is maximum among all the possible edges in $(A\times V)\setminus (E\cup S)$ to be added to $S$. The next theorem shows that $\sigma(A,\cdot)$ is monotone and submodular.

\begin{theorem}\label{th:submodular:MICM}
Given a graph $G=(V, E, p)$, $\sigma(A, S)$ is a monotonically increasing submodular function of sets $S\subseteq (A\times V)\setminus E$.
\end{theorem}

In order to prove the theorem, we first show, in the next lemma, that function $R(A,\cdot)$ is submodular with respect to the insertion of edges outgoing nodes in $A$ to live-edge graphs. Note that this is a deterministic property. 

\begin{lemma}\label{lem:submodular:R}
Given a set of nodes nodes $A\subseteq V$, two live-edge graphs $X$ and $Y$ such that $E_X\subseteq E_Y$ and $E_Y\setminus E_X \subseteq A\times V$, and an edge $e\in (A\times V)\setminus E$, let $X^+$ and $Y^+$ denote the live-edge graphs obtained by adding $e$ to $X$ and $Y$, respectively, that is $X^+=(V,E_X\cup\{e\})$ and $Y^+=(V,E_Y\cup\{e\})$. Then,
$
|R(A,Y^+)| - |R(A, Y)| \leq |R(A,X^+)| - |R(A, X)|.
$
\end{lemma}
\begin{proof}
Let $r(X,e) = R(A,X^+) \setminus R(A,X)$, that is, $r(X,e)$ is the set of nodes that are reachable from $A$ in $X^+$ by means of edge $e$ and that are not reachable in $X$. Similarly, let $r(Y,e) = R(A,Y^+) \setminus R(A,Y)$. Since $E_X\subseteq E_Y$ and $E_Y\setminus E_X \subseteq A\times V$, then $R(A,X)\subseteq R(A,Y)$ and the set of nodes that are reachable from $A$ \emph{only} by means of edge $e$ is smaller in $Y^+$ than in $X^+$. It follows that $|r(X,e)|\geq |r(Y,e)|$. Therefore, $|R(A,Y^+)| - |R(A, Y)| =  |r(X,e)|\geq |r(Y,e)| = |R(A,X^+)| - |R(A, X)|$.
\end{proof}

We can now prove Theorem~\ref{th:submodular:MICM}.
\begin{proof}[Theorem~\ref{th:submodular:MICM}]
To prove that $\sigma$ is a monotonically increasing function, we show that for each $S \subseteq (A\times V)\setminus E$ and $e=(a,v)\in (A\times V)\setminus (E\cup S)$, $\sigma(A,S\cup{e})-\sigma(A, S) \geq 0$, that is
\begin{equation}\label{eq:monotonic}
\sum_{X \in \chi(S\cup\{e\})} \mathbb{P}[X]\cdot|R(A, X)|-\sum_{X \in \chi(S)} \mathbb{P}[X]\cdot|R(A, X)| \geq 0.
\end{equation}
For each live-edge graph $X$ in $\chi(S))$, there are two different corresponding live-edge graph $X^+$ and $X^-$ in $\chi(S\cup\{e\})$, whose edge sets depend on the outcome of the coin flipped for $e$, that is $E_{X^{+}} = E_X\cup\{e\}$ and $E_{X^{-}} = E_X$, respectively. The probabilities for the live-edge graph $X^{+}$ and $X^{-}$ to occur, are: $\mathbb{P}[X^{+}]=\mathbb{P}[X]\cdot p_{e}$ and $\mathbb{P}[X^{-}]=\mathbb{P}[X]\cdot (1-p_{e})$, while the set of reachable nodes are such that $R(A, X)\subseteq R(A, X^{+})$ and $R(A, X) = R(A, X^{-})$, because in $X^{+}$ there is one more edge $e$ and $X^{-} = X$. Therefore, $|R(a, X^{+})|\geq |R(a, X)|$ and we can write Inequality~\eqref{eq:monotonic} as:
\begin{align*}
 \sum_{X \in \chi(S)}\left( \mathbb{P}[X^{+}]\cdot |R(A, X^{+})| + \mathbb{P}[X^{-}]\cdot |R(A, X^{-})|\right) - \sum_{X \in \chi(S)}\mathbb{P}[X]\cdot |R(A, X)| =\\
  \sum_{X \in \chi(S)}\left( \mathbb{P}[X]\cdot p_e\cdot |R(A, X^{+})| + \mathbb{P}[X]\cdot (1-p_e)\cdot |R(A, X^{-})| - \mathbb{P}[X]\cdot |R(A, X)|  \right)\geq
 \end{align*}
\begin{align*}
  \sum_{X \in \chi(S)}\left( \mathbb{P}[X]\cdot p_e\cdot |R(A, X)| + \mathbb{P}[X]\cdot (1-p_e)\cdot |R(A, X)| - \mathbb{P}[X]\cdot |R(A, X)| \right) = 0.
\end{align*}
this shows that $\sigma$ is a monotonically increasing.

To prove the submodularity, we show  that for each pair of sets $S,T$ such that $S \subseteq T \subset (A\times V)\setminus E$ and for each $e=(a,v)\in (A\times V)\setminus T$, the increment in expected number of influenced nodes that edge $e$ causes in $S \cup \{e\}$ is larger than the increment it produces in $T \cup \{e\}$, that is $\sigma(A, S \cup \{e\})-\sigma(A, S) \geq \sigma(A, T \cup \{e\})-\sigma(A, T)$, or
\begin{eqnarray}
 \sum_{X \in \chi(S\cup\{e\})} \mathbb{P}[X]\cdot |R(A, X)|-\sum_{X \in \chi(S)} \mathbb{P}[X]\cdot |R(A, X)| \geq\label{eq:submodular:S} \\
 \sum_{X \in \chi(T\cup\{e\})} \mathbb{P}[X]\cdot |R(A, X)|-\sum_{X \in \chi(T)} \mathbb{P}[X]\cdot |R(A, X)|\label{eq:submodular:T}.
\end{eqnarray}

For each live-edge graph $X$ in $\chi(S)$ let us denote by $\chi(T,X)$ the set of live-edge graphs in $\chi(T)$ that have $X$ as a subgraph and possibly contain other edges in $T\setminus S$. In other words, a live-edge graphs in $\chi(T,X)$ has been generated with the same outcomes as $X$ on the coin flips in the edges of $E\cup S$ and it has other outcomes for edges in $T\setminus S$. Note that $|\chi(T,X)| = 2^{|T\setminus S|}$. As in the proof for monotonicity, for each live-edge graph $X$ in $\chi(T)$, let $X^+$ and $X^-$ be the live-edge graphs in $\chi(T\cup\{e\})$ such that $E_{X^{+}} = E_X\cup\{e\}$ and $E_{X^{-}} = E_X$, respectively. Again, $\mathbb{P}[X^{+}]=\mathbb{P}[X]\cdot p_{e}$, $\mathbb{P}[X^{-}]=\mathbb{P}[X]\cdot (1-p_{e})$, $R(A, X)\subseteq R(A, X^{+})$, and $R(A, X) = R(A, X^{-})$.

Then, Formula~\eqref{eq:submodular:S} is equal to 
\begin{eqnarray*}
&&\sum_{X \in \chi(S)}\left(\mathbb{P}[X^+]\cdot |R(A, X^+)|+\mathbb{P}[X^-]\cdot |R(A, X^-)| -\mathbb{P}[X]\cdot |R(A, X)|\right)=\\
&&\sum_{X \in \chi(S)}\left(\mathbb{P}[X]\cdot p_e\cdot |R(A, X^+)|+\mathbb{P}[X]\cdot(1-p_e)\cdot|R(A, X^-)| -\mathbb{P}[X]\cdot |R(A, X)|\right)=\\
&&\sum_{X \in \chi(S)}\mathbb{P}[X]\cdot p_e \cdot(|R(A, X^+)| -|R(A, X)|).
\end{eqnarray*}
While, Formula~\eqref{eq:submodular:T} is equal to 
\begin{eqnarray*}
&&\sum_{X \in \chi(S)} \sum_{Y \in \chi(T,X)} \left(   \mathbb{P}[Y^+]\cdot |R(A, Y^+)|+\mathbb{P}[Y^-]\cdot |R(A, Y^-)| -\mathbb{P}[Y]\cdot |R(A, Y)| \right) = \\
&&\sum_{X \in \chi(S)} \sum_{Y \in \chi(T,X)} \left(\mathbb{P}[Y]\cdot p_e \cdot(|R(A, Y^+)| -|R(A, Y)|)\right).
\end{eqnarray*}
 
By Lemma~\ref{lem:submodular:R}, $|R(A, Y^+)| -|R(A, Y)| \leq |R(A, X^+)| -|R(A, X)|$. Moreover, $\sum_{Y \in \chi(T,X)} \mathbb{P}[Y] = \mathbb{P}[X]$ and then 
\[
\sum_{X \in \chi(S)} \sum_{Y \in \chi(T,X)} \mathbb{P}[Y]\cdot p_e \cdot(|R(A, Y^+)| -|R(A, Y)|) \leq \sum_{X \in \chi(S))}\mathbb{P}[X]\cdot p_e \cdot(|R(A, X^+)| -|R(A, X)|),
\]
which concludes the proof.
\end{proof}

\begin{algorithm2e}
\caption{\Greedy \MICM algorithm with approximate estimation of marginal increment.}
\label{alg:greedytwo}
\SetKwInput{Proc}{algorithm}
%\Proc{\Greedy}
\SetKwInOut{Input}{Input}
\SetKwInOut{Output}{Output}
\Input{A directed graph $G=(V,E)$; a set of vertices $A\subseteq V$; and an integer $k\in\mathbb{N}$}
\Output{Set of edges $S\subseteq (A\times V)\setminus E$ such that $|S|\leq k$}
$S:=\emptyset$\;
\For{$i=1,2,\ldots,k$\label{alg:greedytwo:for}}
{
  \ForEach{$e \in (A\times V)\setminus (E\cup S)$\label{alg:greedytwo:foreachedge}}
  {
    Use repeated sampling to estimate a $(1+\lambda)$-approximation of $\sigma(A,S\cup \{e\})$ with prob. $1-\delta$\;
    Let $\tilde{\sigma}(A, S\cup \{e\})$ be the estimation\;
  }
  $\hat{e} = \arg\max\{ \tilde{\sigma}(A,  S\cup \{e\})~|~e=(a,v) \in (A\times V)\setminus (E\cup S)\}$\;
  $S := S \cup \{\hat{e}\}$\;
}
\Return $S$\;
\end{algorithm2e}

Theorem~\ref{th:submodular} can be generalized to the case in which each step of the greedy algorithm selects an element whose marginal increment is within a factor $(1 + \lambda)$ to the maximal one. In this case, the greedy algorithm guarantees a $\left(1-\frac{1}{e}-\epsilon\right)$-approximation, where $\epsilon$ depends on $\lambda$ and goes to $0$ as $\lambda\rightarrow 0$. By combining these results, we can formally define algorithm~\ref{alg:greedytwo} that differs from algorithm~\ref{alg:greedy} on how it computes $\sigma(A, S)$.
\begin{theorem}
algorithm~\ref{alg:greedytwo} guarantees an approximation factor of $\left(1-\frac{1}{e}-\epsilon\right)$ for the \MICM problem, where $\epsilon$ is any positive real number.
\end{theorem}

\section{Approximation algorithms for the \CICM problem}\label{sec:approx_easy}
In this section we introduce our approximation algorithm for the \CICM problem.
First we propose a greedy algorithm that achieves an approximation factor of  $\frac{1}{2}\left(1-\frac{1}{e}\right)$ for the \CICM problem, then we improve such approximation factor to $\left(1-\frac{1}{e}\right)$ by using an enumeration technique.

\begin{algorithm2e}
\caption{\Greedy \CICM algorithm.}
\label{alg:greedy_r2}
\SetKwInput{Proc}{algorithm}
%\Proc{\Greedy}
\SetKwInOut{Input}{Input}
\SetKwInOut{Output}{Output}
\Input{A directed graph $G=(V,E)$, an integer $k\in\mathbb{N}$ a seed set $A$}
\Output{A set of edges $S\subseteq (A\times V)\setminus E$ such that $c(S)\leq k$}
$S:=\emptyset$\;
$T:=(A\times V)\setminus E$\;
$e_{M} : = \arg\max_{e\in (A\times V)\setminus E} \left\{ \sigma(A,S\cup\{e\}) \right\}$\;\label{alg:greedy_r2:max}
\While{$T\neq\emptyset$}
{\label{alg:greedy_r2:whilestart}
      $\hat{e} : = \arg\max_{e\in  T} \left\{ \frac{\delta(S\cup\{e\},S)}{c_{e}} \right\}$\;
      \If{$k-c_{\hat{e}}\geq 0$}
      {
        $S := S \cup \{ \hat{e} \}$\;
        $k := k-c_{\hat{e}}$\;
      }
      $T := T \setminus \{ \hat{e} \}$\;
}\label{alg:greedy_r2:whileend}
\Return $\arg\max\{ \sigma(A,S), \sigma(A, \{e_{M}\}) \}$\;
\end{algorithm2e}
Our algorithm, whose pseudocode is reported in algorithm~\ref{alg:greedy_r2}, outputs a solution that maximizes the expected number of affected nodes between two possible solutions described in the following. The first solution is found at line~\ref{alg:greedy_r2:max} and is made of a single edge $(a_{M},v_{M})$ for which $\sigma(A,S\cup\{(a_M,v_M)\})$ is maximized; the second solution is obtained by a greedy algorithm at lines~\ref{alg:greedy_r2:whilestart}--\ref{alg:greedy_r2:whileend}.

In particular, the greedy phase, selects at each step an edge $\hat{e}$ to be added to the solution $S$ obtained at the previous iteration, such that the ratio between $\delta(S\cup \{ \hat{e}\},S)$ and $c_{\hat{e}}$ is maximized. 
Then, if cost $c(S\cup\{\hat{e}\})$ does not violate the budget, the edge $\hat{e}$ is added to $S$, otherwise the edge is discarded.

Next, we analyse the performance guaranteed by algorithm~\ref{alg:greedy_r2}.
We denote by $S^*$ an optimal solution to the problem.
 Let us consider the iterations $i$ executed by the greedy algorithm, for $i\geq 1$, let us denote by $j_i$ the index of such iterations, $j_i<j_{i+1}$. Let $j_l$ be the index of iterations until an edge is added to the solution without exceeding the given budged and let $j_{l+1}$ be the index of the first iteration in which an element in $S^*$ is considered (i.e. it maximizes the above ratio) but not added to $S$ because it violates the budget constraint.
% the element added to $(A,S)$ belongs to $(A^*,S^*)$. For $i\geq 1$, let us denote by $j_i$ the index of such iterations, $j_i<j_{i+1}$, and let $j_{l+1}$ be the index of the first iteration in which an element in $(A^*,S^*)$ is considered but not added to $(A,S)$. That is in all the $l$ iterations $j_1,j_2,\ldots,j_l$ the greedy algorithm adds to $(A,S)$ an element in $(A^*,S^*)$.
We denote by $S_i$ the solution at the end of iteration $j_i$ and by $\bar{c}_i$ the marginal cost of $S_i$ as computed in the above ratio, $\bar{c}_i=c_{\hat{e}}$, where $\hat{e}$ is the edge selected at iteration $i$.

The next lemmas are the core of our analysis, note that the statements are similar to lemmas in~\cite{KMN99}.

\begin{lemma}\label{lem:marginalcosts_r2}
 After each iteration $j_i$, $i=1,2,\ldots,l+1$, $\sigma(A,S_i) - \sigma(A,S_{i-1}) \geq \frac{\bar{c}_i}{k} (\sigma(A,S^*) - \sigma(A,S_{i-1}))$.
\end{lemma}
\begin{proof}
Firs we define $\delta_i$ to be the expected number of nodes affected by solution $S_i$ and not affected by solution $S_{i-1}$, $\delta_i=\delta(S_i,S_{i-1})$.

It is easy to see that the fallowing inequality holds
\begin{equation}\label{eq:one_r2}
 \sigma(A,S^*) - \sigma(A,S_{i-1}) \leq \sum_{e\in S^*\setminus S_{i-1}} \delta(S_{i-1}\cup\{e\},S_{i-1}),
\end{equation}
i.e. the value $\sigma(A,S^*) - \sigma(A,S_{i-1})$ is at most the sum, for each edge in $S^*$ and not in $S_{i-1}$, of the expected number of nodes affected by such edge and not affected by solution $(A,S_{i-1})$.

Since the greedy algorithm selects at each step the element that maximizes the ratio between $\delta_i$ and $\bar{c_i}$, for each  $e\in S^*\setminus S_{i-1}$ the following holds, 
\begin{align*}
 \frac{\delta(S_{i-1}\cup\{e\},S_{i-1})}{c_e} &\leq \frac{\delta_i}{\bar{c}_i}.
\end{align*}
Therefore,
\[
 \sum_{e\in S^*\setminus S_{i-1}} \delta(S_{i-1}\cup\{e\},S_{i-1})\leq \sum_{e\in S^*\setminus S_{i-1}} \frac{\delta_i}{\bar{c}_i}c_e = \frac{\delta_i}{\bar{c}_i} \left( \sum_{e\in S^*\setminus S_{i-1}} c_e\right) \leq k \frac{\delta_i}{\bar{c}_i}.
\]
To conclude the proof, we need to show that $\delta_i =  \sigma(A,S_i) - \sigma(A,S_{i-1})$. Indeed, if $S_i\setminus S_{i-1}=\{e\}$,
\begin{align*}
 \delta_i &= \sum_{X \in \chi(S_i)} \mathbb{P}[X]\cdot\left(|R(A, X)| - |R(A, X^e)|  \right)\\
          &= \sigma(A,S_i) - \sum_{X \in \chi(S_{i-1})} \left(p_e|R(A, X)| + (1-p_e)|R(A, X)|\right)\\
          &= \sigma(A,S_i) - \sigma(A,S_{i-1}).
\end{align*}

\end{proof}

Armed with Lemma~\ref{lem:marginalcosts_r2}, we prove the next lemma by induction on iterations $j_i$.
\begin{lemma}\label{lem:induction_r2}
 After each iteration $j_i$, $i=1,2,\ldots,l+1$, $$\sigma(A,S_i) \geq \left[ 1- \prod_{\ell=1}^i\left( 1 - \frac{\bar{c}_\ell}{k}\right) \right]\sigma(A,S^*).$$
\end{lemma}
\begin{proof}
 For $i=1$, by Lemma~\ref{lem:marginalcosts_r2}, $\sigma(A,S_1)  \geq \frac{\bar{c}_1}{k} \sigma(A,S^*) =\left[ 1- \left( 1 - \frac{\bar{c}_1}{k}\right) \right]\sigma(A,S^*)$. Let us assume that the statement holds for $j_1,j_2,\ldots,j_{i-1}$, then 
 \begin{align*}
  \sigma(A,S_i) &= \sigma(A,S_{i-1}) + \left[ \sigma(A,S_i) - \sigma(A,S_{i-1})\right]\\
                  &\geq \sigma(A,S_{i-1}) + \frac{\bar{c}_i}{k}\left[ \sigma(A,S^*) - \sigma(A,S_{i-1})\right]\\
                  &= \sigma(A,S_{i-1}) \left(1-\frac{\bar{c}_i}{k}\right) + \frac{\bar{c}_i}{k}\sigma(A,S^*)
 \end{align*}
 where the inequalities follows from Lemma~\ref{lem:marginalcosts_r2}.

 To conclude the proof we apply the inductive hypothesis: 
\begin{align*}
             \sigma(A,S_{i-1}) \left(1-\frac{\bar{c}_i}{k}\right) + \frac{\bar{c}_i}{k}\sigma(A,S^*)&\geq \left[ 1- \prod_{\ell=1}^{i-1}\left( 1 - \frac{\bar{c}_\ell}{k}\right) \right]\left(1-\frac{\bar{c}_i}{k}\right)\sigma(A,S^*)+ \frac{\bar{c}_i}{k}\sigma(A,S^*)\\
                  &=\left[ 1- \prod_{\ell=1}^i\left( 1 - \frac{\bar{c}_\ell}{k}\right) \right]\sigma(A,S^*).
 \end{align*}
\end{proof}

\begin{theorem}
Algorithm~\ref{alg:greedy_r2} achieves an approximation factor of  $\frac{1}{2}\left(1-\frac{1}{e}\right)$ for the \CICM problem.
%  If $(A,S)$ is the solution produced by algorithm~\ref{alg:greedy_r2}, then $\sigma(A,S)\geq \frac{1}{2}\left(1-\frac{1}{e}\right)\sigma(A,S^*)$.
\end{theorem}
\begin{proof} 
We first observe two facts:
\begin{enumerate}
\item since $(A,S_{l+1})$ violates the budget, then $c(A,S_{l+1})>k$,
\item for a sequence of numbers $a_1,a_2,\ldots,a_n$ such that $\sum_{\ell=1}^n a_\ell = A$, the fallowing holds: $\prod_{i=1}^n \left(1 -\frac{a_i}{A}\right)\leq \left(1-\frac{1}{n}\right)^n$.
\end{enumerate}
Therefore, by applying Lemma~\ref{lem:induction_r2} for $i=l+1$ we obtain:
\begin{align*}
 \sigma(A,S_{l+1}) & \geq \left[ 1- \prod_{\ell=1}^{l+1}\left( 1 - \frac{\bar{c}_\ell}{k}\right) \right]\sigma(A,S^*)\\
                         & \geq \left[ 1- \prod_{\ell=1}^{l+1}\left( 1 - \frac{\bar{c}_\ell}{c(S_{l+1})}\right) \right]\sigma(A,S^*)\\
                         & \geq \left[ 1- \left( 1 - \frac{1}{l+1}\right)^{l+1} \right]\sigma(A,S^*)\\
                         & \geq \left(1-\frac{1}{e}\right) \sigma(A,S^*).
\end{align*}
% By using arguments similar to those in the last part of Lemma~\ref{lem:marginalcosts_r2}, we can show that:
It follows that:
 \begin{align}
 \sigma(A,S_{l+1}) = \sigma(A,S_{l}) + \delta_{l+1}\geq \left(1-\frac{1}{e}\right) \sigma(A,S^*).\label{eq:simple}
\end{align}
  Moreover, since $\delta_{l+1} \leq \sigma(A, \{e_{M}\})$, we get:
\[
 \sigma(A,S_{l}) + \sigma(A, \{e_{M}\})\geq \left(1-\frac{1}{e}\right) \sigma(A,S^*).
\]
Finally, note that $\max\{ \sigma(A,S_l), \sigma(A, \{e_{M}\}) \}\geq \frac{1}{2}\left(1-\frac{1}{e}\right) \sigma(A,S^*)$.
\end{proof}

We now  propose  an algorithm which improves the performance guarantee of algorithm~\ref{alg:greedy_r2}. 
Let $M$ a fixed integer, we consider all the solutions of cardinality $M$ (i.e. $|S| = M$) which have cost at most $k$, $c(S)\leq k$, and we complete each solution by using the greedy algorithm. The pseudocode is reported in algorithm~\ref{alg:greedy_r2_2}.

\begin{algorithm2e}
\caption{\Greedy \MICM algorithm with enumeration technique.}
\label{alg:greedy_r2_2}
\SetKwInput{Proc}{algorithm}
%\Proc{\Greedy}
\SetKwInOut{Input}{Input}
\SetKwInOut{Output}{Output}
\Input{A directed graph $G=(V,E)$, integer $M\in\mathbb{N}$ and an integer $k\in\mathbb{N}$}
\Output{A set of edges $S\subseteq (A\times V)\setminus E$ such that $c(S)\leq k$}
$S_1 : = \arg\max\{ \sigma(A,S): |S|< M, c(S)\leq k \}$\;\label{alg:greedy_r2_2:first}
$S_2:=\emptyset$\;
$T:=(A\times V)\setminus E$\;
\ForEach{$ S\subseteq T \text{ such that } |S|= M, c(S)\leq k $}
{ 
$T := T \setminus S$\;
Complete $S$ by using algorithm~\ref{alg:greedy_r2} with $T$ as possible edge set\;
%\While{$T\neq \emptyset$ and $k\geq 0$}
%{\label{alg:greedy_r2_2:whilestart}
%      $(\hat{a},\hat{v}) : = \arg\max_{(a,v)\in A\times V\cap T} \left\{ \frac{\delta(S\cup\{(a,v)\},S)}{c_{(a,v)}} \right\}$\;
%      \If{$k-c_{(\hat{a},\hat{v})}\geq 0$}
%      {
%        $S := S \cup \{ (\hat{a},\hat{v}) \}$\;
%        $k := k-c_{(\hat{a},\hat{v})}$\;
%      }
%      $T := T \setminus \{ (\hat{a},\hat{v}) \}$\;
%
%}\label{alg:greedy_r2_2:whileend}
\If{$\sigma(A, S)>\sigma(A, S_2)$}
{
	$S_2 := S$\;	
}
}
\Return $\arg\max\{ \sigma(A,S_1), \sigma(A, S_2)\}$\;
\end{algorithm2e}

\begin{theorem}\label{th:final}
 For $M\geq 3$ algorithm~\ref{alg:greedy_r2_2} achieves an approximation factor of  $\left(1-\frac{1}{e}\right)$ for the \CICM problem.
\end{theorem}
\begin{proof}
We assume that $|S^*|>k$ since otherwise algorithm~\ref{alg:greedy_r2_2} finds an optimal solution.

We sort the edges in $S^*$ in decreasing order according to their marginal increment in the objective function.

Let $S_Z$ be the first $M$ elements in this order. We now consider the iteration of algorithm~\ref{alg:greedy_r2_2} in which element $Z$ is considered. We define $S_{Z'}$ such that $S=S_Z \cup S_{Z'}$, where $S$ is the solution obtained after applying the greedy algorithm. It follows that:
\[
\sigma(A, S)=\sigma(A, S_Z)+\delta(S_Z \cup S_{Z'}, S_Z). 
\]

The completion of $S_Z$ to $S$ is an application of the greedy heuristic from algorithm~\ref{alg:greedy_r2} therefore, we can use the result from the previous theorems.  
Let us consider the iterations executed by the greedy algorithm during the completion of $S_Z$ to $S$. For $i\geq 1$, let us denote by $j_i$ the index of such iterations, $j_i<j_{i+1}$, and let $j_{l+1}$ be the index of the first iteration in which an edge in $S^* \setminus S_Z$ is considered but not added to $S_{Z'}$ because it violates the budget constraint. Applying Inequality~\eqref{eq:simple}, we get:
\begin{align}
\delta(S_Z \cup S_{Z'}, S_Z) + \delta_{l+1}\geq \left(1-\frac{1}{e}\right) \sigma(A, S^* \setminus S_Z)\label{eq:approx_r2}
\end{align}
we observe that, since we ordered the elements in $S^*$, $\delta_{l+1}\leq \frac{\sigma(A, S_Z)}{M}$. 

Therefore, applying Inequality~\eqref{eq:approx_r2} and the previous observation:
\begin{align*}
\sigma(A, S)&=\sigma(A, S_Z)+\delta(S_Z \cup S_{Z'},S_Z)\\
 & \geq \sigma(A, S_Z)+ \left(1-\frac{1}{e}\right) \sigma(A,S^* \setminus S_Z) -\delta_{l+1}\\
& \geq \sigma(A, S_Z)+ \left(1-\frac{1}{e}\right) \sigma(A,S^* \setminus S_Z) -\frac{\sigma(A, S_Z)}{M}\\
&\geq \left(1-\frac{1}{M}\right) \sigma(A, S_Z)+ \left(1-\frac{1}{e}\right) \sigma(A,S^* \setminus S_Z)
\end{align*}
But, $\sigma(A, S_Z)+ \sigma(A,S^* \setminus S_Z)\geq\sigma(A, S^*)$, and we get:
\begin{align*}
\sigma(A, S)&\geq  \left(1-\frac{1}{e}\right) \sigma(A,S^*)+ \left(\frac{1}{e}-\frac{1}{M}\right) \sigma(A,S_Z) \\
&\geq  \left(1-\frac{1}{e}\right) \sigma(A,S^*),\quad \text{ for }M\geq 3
\end{align*}
proving the theorem.
\end{proof}

As in the previous section, Theorem~\ref{th:final} can be generalized to the case in which each step of the greedy algorithm uses repeated sampling to estimate a $(1 + \lambda)$-approximation of $\sigma(A, S)$ with probability $(1-\delta)$. In this case, the greedy algorithm guarantees a $\left(1-\frac{1}{e}-\epsilon\right)$-approximation, where $\epsilon\rightarrow 0$ as $\lambda\rightarrow 0$. 

\section{Conclusion and future research}\label{sec:concl}
In this paper, we have shown that \MICM admits a constant factor approximation algorithm by proving that the expected number of activated nodes is monotonically increasing and submodular with respect to the possible set of edges incident to the seeds. We further provide an upper bound to the approximation factor that is slightly higher than that guaranteed by our algorithm. 
Moreover, we have shown how to approximate the more general \CICM problem: we first provide an algorithm  which achieves a $\frac{1}{2}\left(1-\frac{1}{e}\right)$ approximation factor and then, using the enumeration technique, we improve the performance guarantee within a factor of $1-\frac{1}{e}$.

As future works, we plan to analyze a minimization version of the \CICM problem where we allow the deletion of edges incident to seeds. Moreover, our intent is to study the same problem in a generalization of ICM, which is the Decreasing Cascade model. In this model the probability of a node $u$ to influence $v$ is non-increasing as a function of the set of nodes that have previously tried to influence $v$. Other research directions that deserve further investigation include the study of the \CICM problem on different information diffusion models such as LTM or the Triggering Model~\cite{KKT15}.
We are also interested, as a future work, in studying a generalization of the problem of Kempe et al. in which we are allowed to spend part of the budget to select seeds, which are not given, and part of it to create new edges incident to such seeds. %More formally, we will investigate the problem of selecting a set of initial seeds and adding a small number of edges incident to the seeds, without exceeding a given budget $k$, in order to maximize the spread of information in the network. 
Finally, we plan to assess the performance of our greedy algorithm from the experimental point of view and to propose some heuristics with the aim of improving the efficiency of the algorithm.

\bibliographystyle{plainurl}% the recommended bibstyle
\bibliography{references_journal}

\end{document}